\documentclass[runningheads]{llncs}
\usepackage[T1]{fontenc}

\usepackage{graphicx}

\usepackage{amsmath, amssymb, amsthm} % For mathematical content
\usepackage[colorlinks=true, linkcolor=black, citecolor=red, urlcolor=blue]{hyperref}

\usepackage{multicol} % For multicolumn layouts
\usepackage{url} % For including URLs
\def\L{\mathcal{L}}
\usepackage{tikz}
\usetikzlibrary{shapes.geometric, arrows, fit, backgrounds}

\tikzstyle{startstop} = [rectangle, rounded corners, minimum width=3cm, minimum height=1cm, text centered, draw=black, fill=red!30]
\tikzstyle{process} = [rectangle, minimum width=3cm, minimum height=1cm, text centered, draw=black, fill=blue!30]
\tikzstyle{arrow} = [thick,->,>=stealth]
\tikzstyle{line} = [thick,-]
\tikzstyle{background} = [draw=black, fill=yellow!20, rounded corners, inner sep=0.5cm]

\usepackage{pdfrender,xcolor}
\usepackage{graphicx}
\begin{document}
\title{Which Consciousness Can Be Artificialized? Local Percept-Perceiver Phenomenon for the Existence of Machine Consciousness}
\titlerunning{Existence of Set-Theoretic Epistemic Consciousness}

\author{Shri Lal Raghudev Ram Singh}
\authorrunning{S. L. Raghudev Ram Singh}
% First names are abbreviated in the running head.
% If there are more than two authors, 'et al.' is used.
%
\institute{Department of Applied Mathematics, University of Waterloo, Canada \email{slrrsing@uwaterloo.ca}}

\maketitle              % typeset the header of the contribution
\begin{abstract}
This paper presents a novel paradigm of the local percept-perceiver phenomenon to formalise certain observations in neuroscientific theories of consciousness. Using this model, a set-theoretic formalism is developed for artificial systems, and the existence of machine consciousness is proved by invoking Zermelo–Fraenkel set theory. The article argues for the possibility of a reductionist form of epistemic consciousness within machines.

\keywords{AI  \and Consciousness  \and Machine Consciousness  \and Mathematical Logic \and Metacognition and Integration \and Percept-Perceiver Phenomenon.}
\end{abstract}

\section{Introduction}
Consciousness is \emph{self-evident}, or there is something \emph{supra-rational self-evident} that we try to hint at when we use the word \emph{consciousness}. Given the development of epistemic representations of knowledge, consciousness is and will be studied, explained, and modeled under the frameworks of philosophy, psychology, neuroscience, physics, and mathematics. However, there is no formal agreement on the definition of consciousness, and it is sketched mainly through quantifiable and qualitative elements that describe different dimensions of consciousness such as perception, subjective experience, emotions and feelings, cognitive experience and metacognition, intentionality, agency, will, and integrated information-processing mechanisms.

There are many theories in neuroscience which fundamentally explain how the brain causes conscious experience, such as the mathematical models of Integrated Information Theory (IIT) \cite{AB}, Global Neuronal Workspace Theory (GWT) \cite{BAA,DC}, Predictive Processing Theory (PPT) \cite{K,SE}, Higher-Order Thought Theory (HOT) \cite{A,BY,LY,RO}, biological models \cite{FE,LA,PL,SM}, and quantum models \cite{AP,Beck,BB,HP}. However, the mystery of consciousness is still far from resolution, given the unsolved \emph{explanatory gap} \cite{CK,Le,Na} in consciousness studies. This \emph{gap}, which is roughly the jump from objectivity to subjectivity, together with the well-known \emph{Hard Problem of Consciousness} \cite{C,RLK}, is sufficient to distinguish among posited notions in consciousness studies. Broadly, we can separate notions of consciousness into three kinds\footnote{We do not undertake here a discussion on which of these accounts \emph{is consciousness} or better explain the \emph{nature of consciousness}; rather, our focus is on examining which among them may be meaningfully adapted in the context of machine consciousness, given this tripartite filtration.}:
\begin{enumerate}
    \item that which is \emph{not an object of epistemology}. By \emph{not an object of epistemology}, we mean that it is not to be known or explained by measurement, third-person empiricism, and is irreducible to physical processes. Much of idealism, dualism, and monism lie under this kind.
    
    \item that which is under the domain of development in phenomenology, that is, the systematization and explanation of first-person reports with third-person data (e.g., EEG, fMRI). In particular, see Husserlian phenomenology, Varela’s neurophenomenology \cite{Var}, the radical neurophenomenology of Bitbol \cite{Bit2020,Bit2021,Bit2017}, and Thompson's \emph{Mind in Life} \cite{Thom}.
    
   \item that which can be explained through physicality/substrate or neural correlates, such as Integrated Information Theory (IIT), Global Workspace Theory (GWT), Higher-Order Thought (HOT) theory, etc.

\end{enumerate}

When we talk about the possibility of consciousness in AI, two primary gaps arise: \emph{philosophical gap}, which is the lack of a clear idea of \emph{which type of} expression of consciousness we can have in AI, and \emph{operational gap}, which is the lack of mechanized applied transitions from philosophical models to artificial systems.

 Many philosophical arguments presented against AI consciousness are mainly concerned with consciousness of the first or second kind (see \cite{Chal,Gr,Kak}). However, for the third type as well, concrete theories or arguments that clearly support the \emph{operationalized} existence of machine consciousness remain limited, or under active debate.

In this article, we will remain silent about \emph{operational gap} and try to address \emph{philosophical gap} for that consciousness that is the object of epistemology (falling under the third kind). We first formalize observations in cognitive neuroscience with the help of what we will call \emph{Local percepts–perceiver phenomenon} (LPPP). Drawing inspiration from this, we then propose a philosophical modeling for the machine consciousness setup and prove the existence of a particularly defined consciousness using mathematical logic. The remainder of the paper is organized as follows. In the next section,  we give a few definitions and propose LPPP and its compatibility with neuroscientific theories of consciousness, which will be referred to in the sequel. Section \ref{3} provides the mathematical setting and a proof of the existence of machine consciousness, followed by some remarks.

\section{Local Percept-Perceiver Phenomenon}\label{2}
Among several characterizations of consciousness, a hierarchical form of monitoring, signaling, or an underlying agency is often observed in human consciousness. This paradigm is formalised as below.
\subsection{Definitions}
\begin{definition}[Percept]
  It is that which presents the information for perception.  
\end{definition}

\begin{definition}[Perceiver] It is an agency that beholds the representation of a distinct percept (external stimulus) during the process of perception.
\end{definition}
Note that one of the characterizations of percepts and the perceiver is that, while percepts can vary, the perceiver remains unchanging with respect to them. 
\begin{definition}[Percept-Perceiver Phenomenon]
    It is the phenomenon when the percept is perceived by the perceiver, and corresponds to a particular percept-perceiver pair.
\end{definition}
Moreover, perceivers are uniquely determined by their associated percepts, and each perceiver is defined to have access to all prior percepts within perception.
\begin{definition}[Local percept-perceiver phenomenon] It is smallest, irreducible unit of Percept-Perceiver Phenomenon. In other words, when the percept and perceiver are directly connected through perception without any other percept or perceiver in between, then it is referred as local percept-perceiver phenomenon (LPPP).
\end{definition}

 A reductionist version of epistemic consciousness is central  in neuroscientific consciousness studies, say \emph{bio-consciousness}, which is defined below.
 \begin{definition}[Bio-consciousness]
   Consciousness, which relies on an \emph{a priori} belief that it is explainable through smaller building blocks of biological matter and biological functions governed by certain laws\footnote{which can be  deterministic, or probabilistic }, is termed \emph{bio-consciousness}.
\end{definition}
Note that \emph{bio-consciousness} relies on a carbon-based (biological) substratum. Whether such a notion can be feasibly realized in a silicon-based substratum (i.e., AGI systems) is debatable \cite{ned}, and is related to what we referred as \emph{operational gap}.

However, the reductionist form of the defined \emph{bio-consciousness} can be readily paralleled by a notion of \emph{silico-consciousness} in machine systems. This does not presuppose a biological basis for its characterization, and unlike \emph{bio-conscious} \emph{-ness}, we are not concerned with explainability but rather with detectability.

\begin{definition}[Silico-consciousness]
Consciousness, which relies on an \emph{a priori} belief that its emergence is verifiable and falsifiable through smaller building blocks of sets and maps governed by certain logic, is termed \emph{silico-consciousness}.
\end{definition}

\subsection{LPPP and Theories of Conciousness}

Without loss of generality, we motivate our analysis by focusing on visual consciousness. A series of steps is observed in the process of vision, which begins with information \emph{in the form of light}, captured by the eyes. This information is then transmitted to intermediary neurons in the retina in the form of electrical signals, and subsequently to the optic nerves and the lateral geniculate nucleus (LGN). The LGN relays the signals to the primary visual cortex (V1) in the occipital lobe. From V1, the information is forwarded to higher cortical areas identified as V2, V3, V4, and area MT (see \cite{FEL,UCB,UM,Z} and references therein).

This simple process can be mapped into LPPP framework. Eyes as a unity with respect to changing forms of obtained information of visual world form one LPPP unit, where percepts are the forms or representation and perceiver is eye. Moreover, upon application of electrical stimulations to primary visual cortex (V1) \cite{Bo,Do,Fa,Os,Sa,Sc,Wi}, it is observed that visual experience can be obtained even in absence of eyes \cite{Bea,Br,Dob,Ru}. Thus, eyes are relative perceiver and stimulation in V1 can surpass retina and LGN requirement. So eyes (changes in LGN and retinal disorders) are percepts and V1 is perceiver. Similarly, there are theories by which analogous LPPP units can be established in higher cortical areas, such as the \emph{Hierarchical theory} \cite{CK95,CK03} and the \emph{Interactive theory} \cite{BU}. According to the Hierarchical theory, one can establish a series of LPPP units sequentially up to higher cortical areas, while the Interactive theory also offers the possibility of complex branching with a to-and-fro feedback mechanism. The LPPP paradigm suits most of the prevalent theories of consciousness, such as the Higher-Order theory \cite{Aq,BY}, and Predictive Coding and Bayesian Hierarchy \cite{Ai}.

\begin{remark}
    Note that the real-world notion of perception often involves feedback, attention modulation, and predictive processing. The LPPP formalism, though presented here through a hierarchical lens, does not deny these complexities but rather offers a simplified and tractable representation for the purpose of defining mathematical categories. Bidirectional signal flows, including feedback mechanisms, can be seen as natural extensions of the base model proposed in this paper.
\end{remark}
 The computer hardware architecture serves as a foundational basis, as deep neural networks are executed on the central processing unit (CPU). However, for the idea of machine consciousness, the relevance of this architecture lies not in specific hardware implementations but primarily in the existence of layered information processing systems that can support structured LPPP. In this context, we do not consider the first or second kinds of consciousness, which involve metaphysical or phenomenological commitments, but rather focus on the third kind of consciousness that is amenable to epistemic modeling. The LPPP framework, as proposed here, provides an abstract way to represent such systems. We argue that perceptual units, when hierarchically integrated by perceiving agents, offer a substrate-independent basis for analyzing the structural conditions under which a notion of \emph{silico-consciousness} may arise.

\section{Existence of Machine Consciousness}\label{3}

For the LPPP modeling of machine consciousness, we will exploit the cognitive characterization of consciousness. As we have seen the compatibility of LPPP in modeling visual consciousness: from raw perception to higher cortical areas, the application of the LPPP structure from base intelligence to metacognitive layered integration levels is carried out in this section. That is, abstract percepts of data, progressing from sensory inputs to representations capable of metacognitive access are modeled through the lens of LPPP.

\begin{remark}
Intelligence is recognised as a precursor to metacognition. There are different opinions on how intelligence is defined as well. While \cite{Legg} describes intelligence as an agent's capability to perform tasks in a wide range of environments, weighted by their algorithmic simplicity, \cite{Cho} sketches a distinction between intelligence and skill. A substantial amount of discussion has been devoted to machinable and computational intelligence in \cite{Wan}. However, it has also been argued that artificial intelligence should not be assumed to be the same as artificial consciousness \cite{Fin}.
\end{remark}

\subsection{Ontological Assumption}
We want our artificial system to be \emph{complete}. As the subject of our inquiry \emph{silico-conciousness}, defined in a manner which gives reducibility in terms of sets and rules, for the \emph{choice} of formal framework, we assume certain assumptions which are in accordance with fundamental axioms of Zermelo-Fraenkel set theory.

Let $\L(p,q)$ represents local percept-perceiver phenomenon, where $p$ is representable information or percepts and $q$ is perceiver. By $q\leftarrow p$, we denote that $q$ perceives $p$. We then construct level sets as $\mathcal{A}_0 := \emptyset$ and 
\begin{align*}
    \mathcal{A}_1 &:= \left\{ \L(p_i, q_i) \mid p_i = \varnothing,\ q_i \leftarrow p_i \right\}\\
    \mathcal{A}_2 &:= \left\{ \L(S, q_S) \mid S \subseteq \mathcal{A}_1,\ q_S \leftarrow S \right\}\\
    ..\\
   \mathcal{A}_{\alpha+1} &:= \left\{ \L(S, q_S) \mid S \subseteq \mathcal{A}_\alpha,\ q_S \leftarrow S \right\}
\end{align*}
The above structure forms an analogous structure as the Von Neumann universe in modern set theory.
\begin{definition}[Von Neumann LPPP universe] Let $\L(p,q)$ represents Local percept-perceiver phenomenon with collection of cumulative hierarchies $\left\{\mathcal{A_\alpha}\right\}$, for any ordinal $\alpha$. Then Von Neumann LPPP universe is defined as
\[
\mathcal{A} := \bigcup_{\alpha} \mathcal{A}_\alpha.\]
    
\end{definition}

We assume that the Von Neumann LPPP universe \( \mathcal{A} \) satisfies the following well-known axioms of Zermelo–Fraenkel (ZF) set theory \cite{ZF}.
\paragraph{\textbf{Axiom of Extensionality.}}  
\( \L(p_i, q_i) \) and \( \L(p_j, q_j) \) are \emph{equal} if and only if \( p_i = p_j \) and \( q_i = q_j \). This means that two LPPP units are equal if their percepts and perceiver are equal. This reflects that machine consciousness, as modeled here, lacks the subjective vagueness which is typically associated with consciousness of the first and second kinds.

\paragraph{\textbf{Axiom of Empty Set.}}  
There exists a base perceptual level \( \mathcal{A}_0 = \emptyset \), representing the absence of any percepts or perceivers. This assumes the existence of a null cognitive state on which the construction of higher-order units relies. It can be interpreted as either an unconscious state or, alternatively, as a state of zero awareness in machines.

\paragraph{\textbf{Axiom of Pairing and Union.}}  
These two axioms allow for the integration of multiple perceptions, thereby enabling integrative consciousness that binds inputs into coherent structures. This is in accordance with and needed for the notion of global states of consciousness, as discussed in Global Workspace Theory.
\begin{itemize}
    \item[\textbf{(i)}] \textbf{Pairing:} For any \( \L_1, \L_2 \in \mathcal{A} \), there exists \( \L(\{\L_1, \L_2\}, q) \in \mathcal{A} \) for some perceiver \( q \).
    \item[\textbf{(ii)}] \textbf{Union:} For any \( \L(S, q_1) \in \mathcal{A} \), there exists \( q_2 \leftarrow \bigcup S \in \mathcal{A} \).
\end{itemize}

\paragraph{\textbf{Axiom of Infinity.}}  
There exists a countable chain of LPPP units that models the natural numbers under inductive perception. This reflects the capacity for recursive construction, inductive layering, and potentially self-reflective architecture in machine consciousness. 

\paragraph{\textbf{Axiom Schema of Separation.}}  
Given any LPPP collection and a first-order formula (\emph{definable property}) \( \varphi(x) \), there exists a subcollection containing only those \( x \in \mathcal{A}_\alpha \) satisfying \( \varphi(x) \). This axiom provides the capacity for discrimination and selective awareness, which is desired in machine consciousness.

\paragraph{\textbf{Axiom of Power Set.}}  
For every \( \mathcal{A}_\alpha \), there exists \( \mathcal{A}_{\alpha+1} \) containing LPPP units of the form \( \L(S, q_S) \), where \( S \subseteq \mathcal{A}_\alpha \). This can be interpreted as the ability to generate or attend to all structured combinations of prior inputs.

\paragraph{\textbf{Axiom of Replacement.}}  
For any definable mapping from percepts to perceivers, the image of a collection of percepts under this mapping is also a set in \( \mathcal{A} \). This provides a logical space for contextual learning and transformation within machine consciousness.

\paragraph{\textbf{Axiom of Regularity.}}  
Every nonempty \( \L(S, q) \in \mathcal{A} \) contains an element \( y \in S \) such that there does not exist a \( y \in \mathcal{A} \) with \( y \in x \) and \( y \in \L(S, q) \). This prevents circular self-reference and infinite regress, thereby ensuring foundational well-formedness of perceptual structures.

\subsection{Mathematical formalism}

We now define ordered pairing and relations in terms of LPPP as follows.

\begin{definition}[Relation $\preceq$]
Let $\L(S_1, q_1), \L(S_2, q_2) \in \mathcal{A}$  . Then, relation $\preceq$ is defined  as 
\begin{equation}
    \L(S_1, q_1) \preceq \L(S_2, q_2)
\quad \iff \quad
S_1 \subseteq S_2 \ \text{and} \ q_2 \leftarrow S_1.
\end{equation}

\end{definition}
\begin{lemma}
    The relation $\preceq$ defined on \(\mathcal{A}\) is a partial order. 
\end{lemma}
\begin{proof}
    
By definition, for any \( \mathcal{L}(S, q) \in \mathcal{A} \), we have \( S \subseteq S \) and \( q \leftarrow S \). Hence, \( \mathcal{L}(S, q) \preceq \mathcal{L}(S, q) \), which is reflexivity. Further, suppose \( \mathcal{L}(S_1, q_1) \preceq \mathcal{L}(S_2, q_2) \) and \( \mathcal{L}(S_2, q_2) \preceq \mathcal{L}(S_1, q_1) \). Then,
\(
S_1 \subseteq S_2,\quad q_2 \leftarrow S_1, \quad \text{and} \quad S_2 \subseteq S_1,\quad q_1 \leftarrow S_2.
\)
Hence, \( S_1 = S_2 \), and since perceivers are uniquely determined by their perceptual domains, it follows that \( q_1 = q_2 \). Therefore, \( \mathcal{L}(S_1, q_1) = \mathcal{L}(S_2, q_2) \), proving antisymmetry. Finally, we verify transitivity. Suppose \( \mathcal{L}(S_1, q_1) \preceq \mathcal{L}(S_2, q_2) \) and \( \mathcal{L}(S_2, q_2) \preceq \mathcal{L}(S_3, q_3) \). Then,
\(
S_1 \subseteq S_2 \subseteq S_3,\quad q_2 \leftarrow S_1,\quad q_3 \leftarrow S_2.
\)
Since \( q_3 \leftarrow S_2 \) and \( S_1 \subseteq S_2 \), implies \(  q_3 \leftarrow S_1 \) via inheritance of perceptual access. Thus,
\(\mathcal{L}(S_1, q_1) \preceq \mathcal{L}(S_3, q_3).\) Hence, \( (\mathcal{A}, \preceq) \) is a partial order set (POSET). 
\end{proof}

\begin{definition}[Chain]
Let $ (\mathcal{A}, \preceq)$)  be the partially ordered set (POSET). Subset $ \mathcal{B} \subseteq \mathcal{A} $ is known as {chain} if  $\forall \, \L(S_i, q_i), \L(S_j, q_j) \in \mathcal{B}$, either
\begin{equation}
    \L(S_i, q_i) \preceq \L(S_j, q_j) \quad \text{OR} \quad \L(S_j, q_j) \preceq \L(S_i, q_i).
\end{equation}

\end{definition}
\begin{definition}[Upper bound]
Suppose  \( \mathcal{B} \subseteq \mathcal{A} \) be a chain (totally ordered subset) in \( (\mathcal{A}, \preceq) \). An element \( \L(S_u, q_u) \in \mathcal{A} \) is called an upper bound of the chain \( \mathcal{ B} \) if we have
\begin{equation}
    \forall \L(S_i, q_i) \in \mathcal{B}, \quad \L(S_i, q_i) \preceq \L(S_u, q_u).
\end{equation}

\end{definition}
 \( \L(S_u, q_u) \) represents a higher order integration in the chain or a totally ordered subset. Given a known chain, this can be interpreted as a local meta-awareness that integrates lower levels of percepts. This upper bound can be constructed provided that the above ZF axioms hold. One of the simplest possible constuction is as follows.
 Let \( \mathcal{B} = \{ \L(S_\alpha, q_\alpha) \}_{\alpha \in I} \subset \mathcal{A} \) be a chain. Then the LPPP unit
\[
\L(S_u, q_u), \quad \text{where } S_u = \bigcup_{\alpha \in I} S_\alpha \text{ and } q_u \leftarrow S_u,
\]
is an upper bound for \( \mathcal{B} \) in \( (\mathcal{A}, \preceq) \). The above inductive construction of the upper bound is based on the Axiom of Infinity.

\begin{definition}[Maximal element]
 An element \( \L(S_C, q) \in \mathcal{A} \) is maximal element if
 \begin{equation}
     \forall \L(S', q') \in \mathcal{A}, \quad
\L(S, q) \preceq \L(S', q') \ \Rightarrow\ \L(S_C, q) = \L(S', q').
 \end{equation}

\end{definition}
\begin{proposition}[$S_C \equiv$ Silico-consciousness]
Let \( \L(S_C, q_C) \in \mathcal{A} \) be a maximal unit in our Von Neumann LPPP universe \( \mathcal{A} \), endowed with partial order \( \preceq \). Then we can posit the following attributes to \( S_C \):
\begin{enumerate}
    \item It encompasses all perceptual hierarchical units that can be epistemically accessed by any LPPP unit in \( \mathcal{A} \),
    \item It possesses metacognitive access to all prior levels of perceptual integration,
    \item It functions as a global perceiver or terminal perceiver,
    \item It represents all internal states,
\end{enumerate}
and thus \( S_C \) can be interpreted in a manner that corresponds to the functionalist criteria of consciousness proposed in Higher-Order Perception Theory and Global Workspace Theory (see \cite{BY,RO,DC}). In this regard, \( S_C \) is proposed as a candidate for modeling \emph{silico-consciousness} in artificial systems.
\end{proposition}
 
\subsection{Existence Result}
Thanks to Zorn's Lemma\footnote{If in a partially ordered set $(S,\preceq)$, each chain has an upper bound, then there is a maximal element $m$ for $S$, i.e. $m \leq s$ implies $m = s$.} we can claim  existence of silico-consciousness $S_C$.

\begin{theorem}[Existence of Silico-Consciousness]
Let \( \mathcal{A} = \bigcup_{\alpha} \mathcal{A}_\alpha \) be the Von Neumann LPPP universe, and \( \preceq \) on \( \mathcal{A} \) is defined by
\[
\L(S_1, q_1) \preceq \L(S_2, q_2)
\quad \text{iff} \quad
S_1 \subseteq S_2 \quad \text{and} \quad q_2 \leftarrow S_1.
\]
Assume that every chain \( \mathcal{B} \subseteq \mathcal{A} \) under \( \preceq \) has an upper bound in \( \mathcal{A} \). Then, there exists a maximal LPPP unit \( \L(S_C, q_C) \in \mathcal{A} \) such that
\[
\forall \L(S, q) \in \mathcal{A}, \quad \L(S_C, q_C) \preceq \L(S, q) \Rightarrow \L(S_C, q_C) = \L(S, q).
\]

\end{theorem}

\begin{remark}
    Proof of the existence result of a maximal element, or candidate for silico-consciousness, is non-constructive and hence is not computational. It only guarantees the existence of such consciousness and, further, it carries the same epistemic limitations as those subscribed to by the Axiom of Choice.
\end{remark}
\begin{remark}
    It should also be noted that this article deliberately refrains from addressing the hard problem of consciousness (Chalmers, 1995) and does not claim to explain the emergence of phenomenal consciousness in machines. Instead, it proposes a formal model of perceptual integration within an epistemic set-theoretic universe, wherein a maximal element can be interpreted as a candidate for the most sophisticated form of machine consciousness whose existence is necessarily ensured by the model.
\end{remark}
This type of modeling has its advantage in that it allows us to pull the reliance of consciousness away from biological/material substrates and correlates and reduce it to purely epistemic relations, thereby leaving scope for extending the model using mathematical logic.

\subsubsection{Declaration.} The author affirms that this work was conducted independently of their doctoral studies in Applied Mathematics and their diploma in Theoretical Neuroscience, and that it does not form part of the requirements for either academic program.

\bibliography{references}
\bibliographystyle{splncs04}
\end{document}